\documentclass[11pt]{article}
\usepackage{color}
\usepackage{epsfig}
\usepackage{subfig}
\usepackage{mathrsfs}
\usepackage{booktabs}
\usepackage{graphicx}
\usepackage{epstopdf}
\usepackage{multirow}
\usepackage{amsmath,amssymb}
\usepackage{amsthm}
\usepackage{fancyhdr}
\usepackage{tabularx}
\usepackage{xfrac}
\usepackage{mathptmx}
\usepackage{mathrsfs}
\usepackage{graphics}
\usepackage{times}

\usepackage[noend]{algpseudocode}
\usepackage{algorithmicx,algorithm}

\usepackage{bm}

\newtheorem{theorem}{Theorem}[section]

\newtheorem{remark}{Remark}[section]

\begin{document}\captionsetup[figure]{labelfont={bf},name={Fig.},labelsep=period}

\title{Feature-Prescribed Iterative Learning Control of Waggle Dance Movement for Social Motor Coordination in Joint Actions}
\author{Bowen Guo, Chao Zhai~\thanks{
Bowen Guo and Chao Zhai are with School of Automation, China University of Geosciences, Wuhan 430074 China,
and with Hubei Key Laboratory of Advanced Control and Intelligent Automation for Complex Systems and Engineering Research Center of Intelligent Technology for Geo-exploration, Ministry of Education, Wuhan 430074 China. [Correspondence: zhaichao@amss.ac.cn] }
}

\maketitle

\begin{abstract}
Extensive experiments suggest that motor coordination among human participants may contribute to social affinity and emotional attachment, which has great potential in the clinical treatment of social disorders or schizophrenia. Mirror game provides an effective experimental paradigm for studying social motor coordination. Nevertheless, the lack of movement richness prevents the emergence of high-level coordination in the existing one-dimensional experiments. To tackle this problem, this work develops a two-dimensional experimental paradigm of mirror game by playing waggle dance between two participants. In particular, an online control architecture of customized virtual player is created to coordinate with human player. Therein, an iterative learning control algorithm is proposed by integrating position tracking and behavior imitation with prescribed kinematic feature.
Moreover, convergence analysis of control algorithm is conducted to guarantee the online performance of virtual player. Finally, the proposed control strategy is validated by matching experimental data and compared with other control methods using a set of performance indexes.
\end{abstract}

\section{Introduction}

Patients with social disorders or schizophrenia inevitably suffer from a lot of pains in daily life, which has brought great misfortune to themselves and their relatives~\cite{Bor07,Cou06}. Existing treatments for social motor disorders mainly have recourse to psychotherapy (i.e., talk therapy) and medication (e.g., drugs or surgery), which are subject to visible downsides such as relatively high treatment cost, side effects of some medication treatments, lack of flexibility and personalization, and inaccessibility of psychotherapy to patients in the remote area.

Similarity theory in social psychology suggests that people prefer to interact with others with similar morphological and behavioral characteristics, and they tend to coordinate their actions unconsciously~\cite{Fol82,Sch14,Pax17}. Recent studies have demonstrated that interpersonal coordination processes are closely related to psychological bonding, and that motor synergy among individuals is expected to
promote social affinity and emotional attachment. As a result, it is promising to treat patients with social disorders by promoting social interaction and enhancing their coordination level through motor synergy~\cite{Wal15,Fen16,Raf15}, which provides a completely different type of noninvasive therapy. To this end, a virtual player (VP) can be designed to interact and coordinate with the patient in joint actions. By prescribing and adjusting the kinematic feature of VP to maintain a certain degree of coordination, the patient is expected to gradually accomplish the transition from unhealthy state to normal state. This helps to achieve the customized therapy and reduce the treatment cost. For instance, the EU Projects ``Alterego"~\cite{alte} and ``Sharespace"~\cite{shar} have been launched to study social motor coordination between the VP and human beings (especially the patients) in joint actions, which allows human participant in remote locations to interact with the avatar in a virtual sensorimotor space.

In practice, the emergence of high-level motor coordination between the VP and HP relies on the effective control of end effector of VP in joint actions. So far, several control strategies have been developed to foster social motor coordination among participants.
For instance, \cite{Zhai18b}~proposes a control architecture for VP with the aid of PD control to generate the personalized movements.
It is observed that the level of motor synergy between the HP and VP is elevated if the kinematics of VP are similar to those of its human partner. \cite{Zhai21} considers time delays of human motor system in transmitting, processing and analyzing visual signals and thus designs a delay controller. Experimental results demonstrates that the appropriate delays may reduce the tracking error and contribute to motor learning of different kinematic feature. In addition to PD control, adaptive control and optimal control approaches are employed to drive a VP for motor coordination with the HP~\cite{Zhai18a,plos16}. Adaptive control is able to achieve the bounded tracking of HP position, but it fails to allow for the desired kinematic feature. To address this issue, optimal control approach is proposed to integrate temporal correspondence with individual motor signature. Besides, reinforcement learning and supervised learning algorithms are also developed to train artificial agents in order to interact and coordinate their movements with HPs~\cite{Aul22,Lom18}. Although the above learning-based approaches are able to accomplish the mission of motor coordination without relying on the dynamics of VP, it takes much training time and requires extensive training data,
which renders it difficult to converge in a short time. Moreover, such learning-based approaches fail to account for the underlying motion neutral process and biological mechanism of social motor coordination in joint actions (e.g., the integration of movement correction into motor memory~\cite{ste05,bra96}).


As a dynamic process of mutual learning and adaptation, a high level of social motor coordination is closely related to the richness and variety of joint actions. Due to spatial constraints, one-dimensional joint actions largely restrict the emergence of complex and ingenious motor coordination among participants, which may weaken motor rehabilitation and emotional resonance of patients suffering from social deficiencies. To overcome
the barrier of spatial constraints, it would be desirable to develop a two-dimensional experimental paradigm of mirror game. In particular,
experimental results demonstrate that human participants prefer the waggle dance movements (i.e., the motion path of ``$\infty$")
to other motion paths in two-dimensional mirror game~\cite{Kas18}. In terms of control and learning strategies, existing approaches either ignore the training process to learn motor program in joint actions or fail to refine control effect of repetitive motion paths for priming collaborative movements by incorporating its repetitive nature. Motivated by the above concerns, this work aims to develop an interactive control architecture for motor therapy of social deficiencies in two-dimensional space. In brief, key contributions of this work are listed as follows.
\begin{enumerate}
  \item Propose a two-dimensional experimental paradigm of mirror game and the corresponding coordination model, which allows to produce more complex and rich collaborative movements, thereby enhancing social attachment between participants and contributing to the motor recovery in patients with social disorders.
  \item Design a customized iterative learning control algorithm for VP to learn and adapt to the motion of human partner in joint actions while possessing the desired kinematic feature and guaranteeing the convergence of state error in a few iterations, which helps to foster motor coordination by incorporating error information into the control for subsequent iterations.
  \item Create a comprehensive set of performance indexes that enables quantifying both temporal correspondence and coordination level, together with the quality of motion priming, and thus demonstrate the superiority of iterative learning control as compared with existing control strategies through experimental validations.
\end{enumerate}

The outline of this paper is given as follows. Section \ref{sec:pre} introduces the preliminaries, which include two-dimensional (2D) mirror game, the coordination model and controller design. Section~\ref{sec:cont} presents the control architecture of VP in joint actions. Section~\ref{sec:theory} focuses on the stability of iterative learning control (ILC) algorithm. Then, experimental validations are conducted using the proposed metrics in Section~\ref{sec:result}. Finally, Section~\ref{sec:con} makes a conclusion and discusses future directions.

\section{Preliminaries}\label{sec:pre}

This section initially introduces the experimental paradigm of 2D mirror game, followed by the construction of the coordination model of VP in joint actions, and finally develops an ILC law with the prescribed kinematic feature.

\subsection{Two-dimensional mirror game}\label{2.1}

\begin{figure}[t!]
\centering
\begin{minipage}[t]{0.48\textwidth}
\centering
\includegraphics[width=7cm,height=4.8cm]{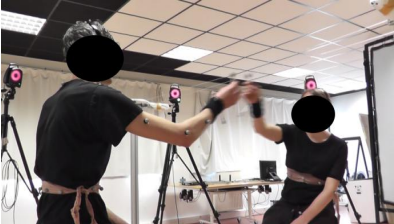}
\end{minipage}
\begin{minipage}[t]{0.48\textwidth}
\centering
\includegraphics[width=7cm,height=4.8cm]{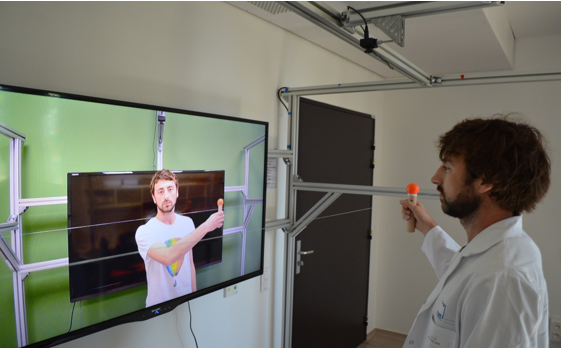}
\end{minipage}
\caption{HP-HP and HP-VP interactions in the mirror game~\cite{Zhai18a}. The left subfigure shows mirror game between two HPs, while the right one displays mirror game between HP and VP.}
\label{fig1}
\end{figure}

Mirror game provides a simple and effective experimental paradigm in the study of social motor coordination between two participants~\cite{Noy11} (see Fig.~\ref{fig1}). In the mirror game, each participant holds the handle of the ball to move the ball back and forth along the string, and the goal is to synchronize their movement and create interesting motions to enjoy the game. Different experimental conditions are designed to investigate social motor coordination in one-dimensional (1D) mirror game. This work extends the original 1D mirror game to its 2D version in order to enrich collaborative movements and foster emotional resonance (see Fig.~\ref{fig2}). Instead of moving along the straight line, participants are instructed to generate the  waggle dance movement in the 2D space. Inspired by 1D mirror game, two experimental conditions for 2D mirror game are described as follows.
\begin{enumerate}
\item  {\it Solo Condition}: This is an individual round. Each participant tries to draw the motion path $"\infty"$ (i.e., waggle dance movement) around a center on his/her own and keeps the center stable while maintaining a certain degree of motion smoothness.
\item  {\it Leader-Follower Condition}: This is a collaborative round. One participant draws the motion path $"\infty"$ around a center and leads the game while maintaining a certain degree of motion smoothness, while the other tries to follow the leader's movement and create synchronized motions and keep the center stable while maintaining a certain degree of motion smoothness.
\end{enumerate}

\begin{figure}[t!]
\scalebox{0.25}[0.25]{\includegraphics{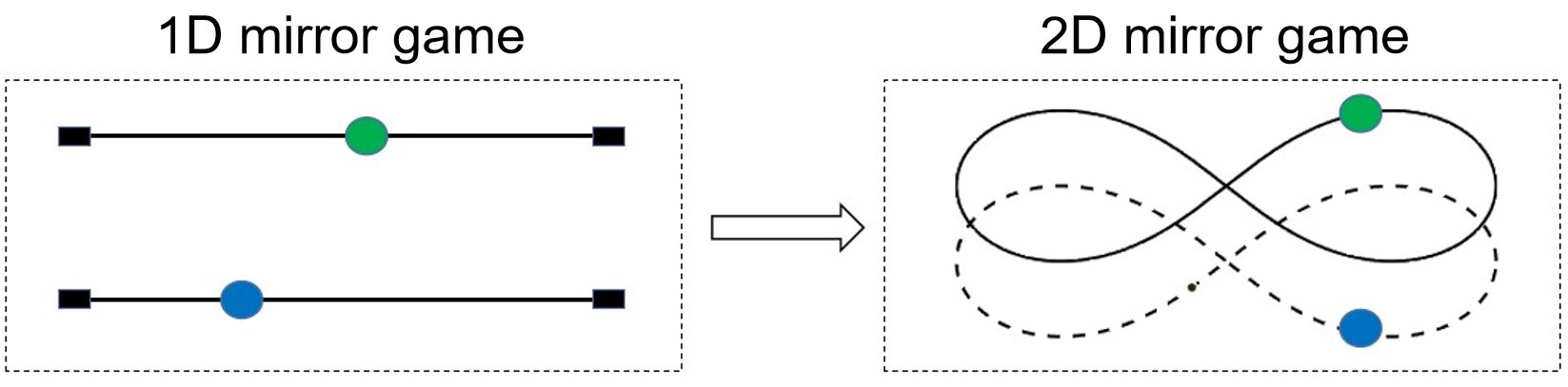}}\centering
\caption{\label{fig2} Schematic diagram on the transition from the 1D mirror game to its 2D version. The blue and green balls are moved by two participants in the mirror game, respectively.}
\end{figure}

\subsection{Coordination model for the 2D joint actions}\label{hkb}

As known, the Haken-Kelso-Bunz (HKB) model is able to capture the phase transition in motor coordination of human end effectors~\cite{Hak04}. Consequently, the coupled HKB oscillators are adopted to describe the 2D motion of end effector in this work.
Then, the 2D coordination model is designed as follows.
\begin{equation}\label{hkb-model}
\left\{
\begin{aligned}
\ddot{z}_1+(\alpha \dot{z}_1^2+\beta z_1^2-\gamma)\dot{z}_1+\omega^2z_1=u_1 \\
\ddot{z}_2+(\alpha \dot{z}_2^2+\beta z_2^2-\gamma)\dot{z}_2+\omega^2z_2=u_2 \\
\end{aligned}
\right.
\end{equation}
where $z_1$, $\dot{z}_1$ and $\ddot{z}_1$ represent the position, velocity and acceleration of VP on the x-axis, respectively.
Likewise, $z_2$, $\dot{z}_2$ and $\ddot{z}_2$ refer to the position, velocity and acceleration of VP on the y-axis, respectively.
Moreover, the parameters $\alpha$, $\beta$, $\gamma$ and $\omega$ characterize the response of uncoupled end effectors when subjected to some reference signals. In particular, $u_1$ and $u_2$ refer to control inputs of the coordination model. Essentially, the objective of this work is to coordinate the motion between HP and VP with the prescribed kinematic characteristic by designing feedback controllers for the VP~\cite{Slo16}.
By introducing new state variables $x_1=z_1$, $x_2=\dot z_1$, $x_3=z_2$, $x_4=\dot z_2$, the control system (\ref{hkb-model}) can be converted
to a state space model as follows
\begin{equation*}\label{state-space}
\left\{
\begin{aligned}
\dot{x}_1 &= x_2 \\
\dot{x}_2 &=-(\alpha x_2^2+\beta x_1^2-\gamma)x_2-\omega^2 x_1+u_1 \\
\dot{x}_3 &= x_4 \\
\dot{x}_4 &=-(\alpha x_4^2+\beta x_3^2-\gamma)x_4-\omega^2 x_3+u_2,\\
\end{aligned}
\right.
\end{equation*}
and it can be described in compact form below
\begin{equation}\label{hkb}
\left\{
\begin{aligned}
&\dot{\bf{x}}=\bf{H}(\bf{x},\bf{\xi})+\bf{B}\bf{u} \\
&\bf{y}=\bf{C}\bf{x},
\end{aligned}
\right.
\end{equation}
with $\mathbf{x}=(x_1,x_2,x_3,x_4)^T$, control input $\mathbf{u}=(u_1,u_2)^T$ and
$$
\mathbf{H}(\mathbf{x},\bf{\xi})=\begin{pmatrix}
    x_2 \\
   -(\alpha x_2^2+\beta x_1^2-\gamma)x_2-\omega^2 x_1 \\
   x_4  \\
   -(\alpha x_4^2+\beta x_3^2-\gamma)x_4-\omega^2 x_3 \\
  \end{pmatrix}, \quad B=\begin{pmatrix}
   0 & 0 \\
   1 & 0 \\
   0 & 0 \\
   0 & 1 \\
 \end{pmatrix}, \quad C=\begin{pmatrix}
   1 & 0 & 0 & 0\\
   0 & 0 & 1 & 0\\
 \end{pmatrix}
$$
where $\bf{\xi}=(\alpha,\beta,\omega,\gamma)^T$ denotes a parameter vector.

\subsection{Iterative learning control with prescribed kinematic feature}\label{ilc-law}
Given the iterative learning process of joint actions in the 2D mirror game, this work adopts the ILC strategy to regulate the movement
of VP~\cite{bri06}. The previous studies show that movement observation may lead to the formation of a lasting specific memory trace in
movement representations, a kind of motor memory~\cite{ste05}, which provides the physiological account of leveraging ILC for motor coordination.
The basic idea is to use the tracking error to correct the control input, which allows to produce a new control input for the next iteration. With the increase of iterations, the ILC strategy enables to generate a recursive sequence of control inputs, which drive
the VP to interact with the HP in the collaborative condition. By taking into account individual kinematic feature and mutual learning and adaptation,
the ILC law is constructed as follows
\begin{equation}\label{ilc}
\mathbf{u}_{k}=\mathbf{u}_{k-1}+\kappa_p \mathbf{e}_{k-1}+\kappa_v \mathbf{\dot{e}}_{k-1}+\kappa_s\mathbf{s}_{k-1},
\end{equation}
where the subscript $k$ denotes the number of iterations, $\kappa_p,\kappa_v,\kappa_s$ are the control gains. In addition, $\mathbf{e}_{k-1}=\mathbf{y}_h-\mathbf{y}_{k-1}$ refers to the tracking error vector between HP and VP, and $\mathbf{y}_h$ refers to the 2D position vector of HP. In particular, $\mathbf{s}_{k-1}=\mathbf{v}-\mathbf{y}_{k-1}$ characterizes the mismatch of kinematic characteristic between the VP and the prescribed feature, and $\mathbf{v}$ represents the pre-recorded velocity trajectory of HP in the solo condition. Based on the resetability of initial state, the state error between VP and HP is bounded under the effect of ILC. By replacing $\mathbf{u}$ in (\ref{hkb}) with $\mathbf{u}_k$ in (\ref{ilc}), the system dynamics becomes
\begin{equation}\label{ilc-sys}
\left\{
\begin{aligned}
&\bf{\dot{x}}_k=\bf{H}(\mathbf{x}_k,\bf{\xi})+\bf{B}\bf{u}_k \\
&\bf{y}_k=\bf{C}\bf{x}_k,\\
\end{aligned}
\right.
\end{equation}
where $\mathbf{x}_k$, $\mathbf{u}_k$ and $\mathbf{y}_k$ are the corresponding state variables, control inputs and control outputs of the system in the $k$-th iteration.

\section{Control Architecture}\label{sec:cont}
This section presents the control architecture of VP, which includes information collection of HP states, ILC with prescribed kinematic feature
and movement generation of VP, as shown in Fig.~\ref{fig3}.
\begin{enumerate}
  \item \textbf{Information collection}: The camera or position sensor is used to collect and record position-time series of HP. After filtering and velocity estimation, it is used as the state information of HP.
  \item \textbf{ILC with prescribed feature}: The position and velocity information of HP is sent to VP, and then the state error is obtained. The iterative learning controller enables to correct the output of VP using the tracking error. In addition, the individual kinematic features is incorporated into the ILC law to generate the personalized movements.
  \item \textbf{Movement generation}: The ILC input is fed into the coordination model of VP, which enables to produce the collaborative movements in order to facilitate social motor coordination between the VP and the HP.
\end{enumerate}

\begin{figure}[t!]
\scalebox{0.075}[0.075]{\includegraphics{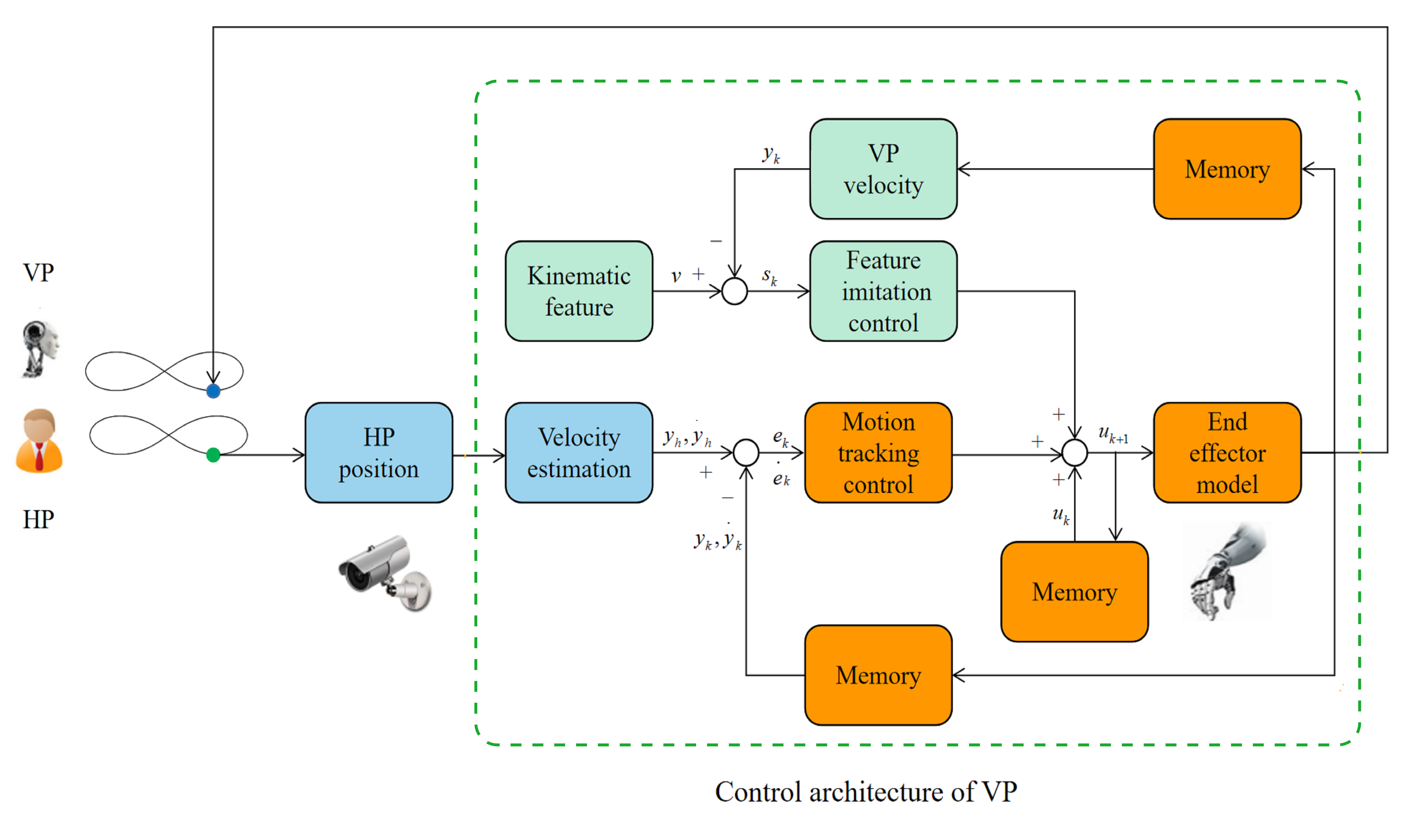}}\centering
\caption{Control Architecture of VP in the framework of ILC}\label{fig3}
\end{figure}

In the control architecture of VP, the vision sensor ensures the information collection of HP, and the ILC law enables VP to
coordinate with the HP in a personalized manner. Algorithm~\ref{tab:oma} elaborates the operation process of ILC algorithm.
Define VP position and HP position as $\mathbf{p}_k(t)=(x_{1,k},x_{3,k})^T$ and $\mathbf{p}_h(t)=(x_{1,h},x_{2,h})^T$.
First of all, the initial states of VP, the threshold of error rate $\epsilon_{th}$ and running time $T$ are specified.
Then,  the rate of position error $\epsilon(t)$ between VP and HP is computed to determine whether it exceeds the specified threshold.
If the condition is satisfied, the control input is updated according to (\ref{ilc}), and it drives the coordination model to generate trajectory of VP. Next, the above iteration process does not stop until the running time is up.
The iteration process allows VP to continuously learn and adapt to HP motion, while maintaining individual kinematic feature.
As the iteration progresses, the VP gradually adapts to the social interaction with the HP, and the state error gets small
as well.
\begin{algorithm}[t]
\caption{\label{tab:oma} Iterative Learning Control Algorithm}
\hspace*{0.02in} Initialize $z_k(0)$, $\dot{z}_k(0)$, $\epsilon_{th}$, $T$, $k=0$
\begin{algorithmic}[1]
\While{$t<T$}
   \State Compute $\varepsilon(t)=\|\mathbf{p}_h(t)-\mathbf{p}_k(t)\|/\|\mathbf{p}_h(t)\|$
   \If {$\varepsilon(t)>\epsilon_{th}$}
       \State Compute tracking error $\mathbf{e}_k(t)$
       \State Compute control input $\mathbf{u}_k(t)$ with (\ref{ilc})
       \State Update k=k+1
    \Else
       \State Obtain control input $\mathbf{u}(t)=\mathbf{u}_k(t)$
       \State Generate the trajectory of VP with (\ref{hkb-model})
       \State Reset $k=0$
    \EndIf \State {\bf end~if}
\EndWhile \State {\bf end~while}
\end{algorithmic}
\end{algorithm}

\section{Theoretical analysis} \label{sec:theory}
This section provides theoretical analysis on the convergence of ILC algorithm and the stability of closed-loop system. Firstly, inequality (\ref{lipschitz}) and inequality (\ref{gronwallinequality}) are presented based on Lipschitz condition and Gronwall lemma, respectively. And then the convergence of state error is obtained. Finally, the above conclusion is extended to the linear case.

\subsection{Convergence analysis of ILC algorithm}\label{sec4.1}

The HKB oscillator is often used to describe interpersonal motor coordination of multiple people, and it can also generate uncoupled  individual motion. Note that human motion is limited, his/her position-time series is bounded. In addition, the moving average filter is proposed to  increase smoothness after collecting HP position trajectory. Meanwhile, HP velocity trajectory is estimated from position trajectory according to the difference method. It is easy to get that HP velocity-time series is bounded, too. Although it is nonlinear, HKB oscillator can  stably generate continuous state trajectory, and it has been proved that HKB oscillator has a limit cycle according to \cite{Zhai18a}, so the output of HKB oscillator is a bounded trajectory. By introducing the appropriate correction value $u_h$, HP motion trajectory can be described by HKB system (\ref{hkb}). In other words, there exists the only expected control input $u_h$ so that VP state matches HP state, that is
\begin{equation*}
\left\{
\begin{aligned}
&\mathbf{\dot{x}}_h=\mathbf{H}(\mathbf{x}_h,\bf{\xi})+\bf{B} \mathbf{u}_h \\
&\mathbf{y}_h=\mathbf{C}\mathbf{x}_h,\\
\end{aligned}
\right.
\end{equation*}
where $\mathbf{x}_h=(x_{1,h},\dot x_{1,h},x_{2,h},\dot x_{2,h})$ and $x_k(0)=x_h(0)=0$. For convenience, the following definitions are given as
\begin{equation} \label{error-sys}
\left\{
\begin{aligned}
&\delta \mathbf{u}_k=\mathbf{u}_h-\mathbf{u}_k  \\
&\delta \mathbf{x}_k=\mathbf{x}_h-\mathbf{x}_k  \\
&\delta \mathbf{H}(\mathbf{x_k} ,\bf{\xi})=H(\mathbf{x_h} ,\bf{\xi})-H(\mathbf{x_k} ,\bf{\xi}).\\
\end{aligned}
\right.
\end{equation}
Note that function $\mathbf{H}(\mathbf{x_k} ,\bf{\xi})$ is Lipschitz continuous, hence one has
\begin{equation*}
\|\mathbf{H}(\mathbf{x}_h,\mathbf{\xi})-\mathbf{H}(\mathbf{x_k},\bf{\xi})\|\leq C_H\|x_h-x_k\|,
\end{equation*}
where
$$
\frac{\partial H}{\partial x}=\begin{pmatrix}
   0 & 1 & 0 & 0\\
   -2\beta x_1x_2-\omega^2 & -3\alpha x_2^2-\beta x_1^2-\gamma &  0 & 0\\
   0 & 0 & 0 & 1\\
   0 & 0 &-2\beta x_3x_4-\omega^2 & -3\alpha x_4^2-\beta x_3^2-\gamma\\
 \end{pmatrix}
$$
and $C_H = \max\limits_{t\in[0,T]}\|\frac{\partial H}{\partial x}\|$. Note that there exists a limit cycle in the HKB oscillator according to \cite{Zhai18a}. As the system runs, the system state tends to the limit cycle, so the system state is bounded. It can be concluded that $\|{\partial H}/{\partial x}\|$ is bounded. Considering that $H(\mathbf{x_k},\bf{\xi})$ is Lipschitz continuous, one obtains
\begin{equation}\label{lipschitz}
\begin{aligned}
\langle H(\mathbf{x_h},\bf{\xi})-H(\bf{x}_k,\bf{\xi}),x_h -x_k\rangle&=[H(\bf{x}_h,\bf{\xi})-H(\bf{x}_k,\bf{\xi})]^T[x_h -x_k] \\
&\leq \|\delta \bf{H}(\bf{x}_k,\bf{\xi})\|\cdot\| \mathbf{x}_h- \mathbf{x}_k \| \\
&\leq C_H\|\mathbf{x}_h-\mathbf{x}_k\|\cdot\|\mathbf{x}_h-\mathbf{x}_k\|=C_H\|\mathbf{x}_h-\mathbf{x}_k\|^2.
\end{aligned}
\end{equation}

\begin{theorem}\label{theo4.1}
For the control system (\ref{ilc-sys}) with control input (\ref{ilc}), the state error between two players is bounded.
\end{theorem}

\begin{proof}
See Appendix~A.
\end{proof}

\begin{remark}
The convergence of state error ensures that VP can follow HP motion. However, the target of VP is not to completely track HP motion. Therefore, solo motion is injected into the control law to ensure that VP exhibit the personal kinematic feature during the interaction. By adjusting $\kappa_s$, VP can achieve the best performance between tracking HP motion and following desired kinematic feature to better match the HP-HP pair. Regardless of iterations, the ILC law degenerates PD control.

\end{remark}


\section{Experimental Validation}\label{sec:result}
In order to validate the proposed control approach, experiments are conducted with leap motion devices. The performance indexes such as root mean square error (RMSE) and circular variance (CV) are adopted to quantify the coordination performance. In addition, a new metric of spatial variation of motion (SVM) is designed to quantify the motion priming. In the HP-VP interaction, HP may unconsciously stop moving or shake suddenly, which results in the movement fluctuation of VP and ultimately affect social motor coordination.
The VP is more popular due to its smooth movements, which is more likely to make people socially attached~ \cite{Kas18}.
Thus, a moving average filter is added to the control architecture of VP to preprocess HP motion trajectory to ensure its smoothness.
In what follows, experimental setup is illustrated, followed by the introduction of performance indexes. Afterwards, statistical analysis is presented with the comparisons of other control strategies.

\subsection{Experimental setup}\label{sec:leap}
The leap motion controller is equipped with dual cameras, like the human eye and is able to coordinate positioning of space objects. Using the leap motion device for experiments can effectively eliminate friction and make the interaction process between participants and VP smoother. Figure~\ref{fig4} shows that human participant interacts with VP on the screen through the leap motion device.
The experiments are carried out in the Unity environment~\cite{gbw2}. The motion data of participants are recorded by the leap motion camera and the script attached to the virtual arm. In the experiments, 8 participants are selected and their motion trajectories are collected under the solo condition in advance to obtain their kinematic feature. Then these participants are classified into $4$ dyads. Each dyad takes for 5 trials, and each trial lasts for $30$ seconds. Each dyad plays the 2D mirror game in the leader-follower condition (LFC). Two participants in each dyad take the role of leader and follower in turn, and the motion data in the script is collected and saved for statistical analysis.

\begin{figure}[t!]
\centering
\begin{minipage}[t]{0.48\textwidth}
\centering
\includegraphics[width=6cm,height=4cm]{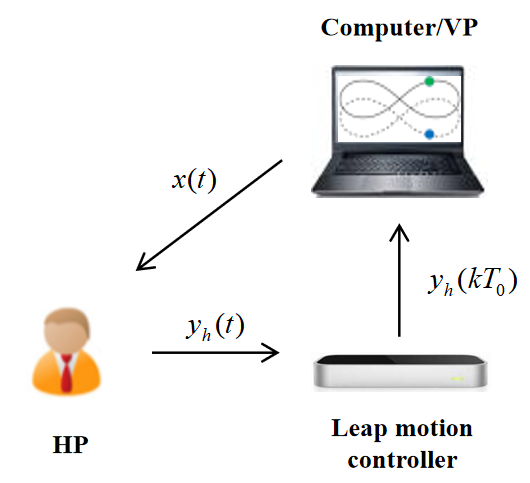}
\end{minipage}
\begin{minipage}[t]{0.48\textwidth}
\centering
\includegraphics[width=6cm,height=4cm]{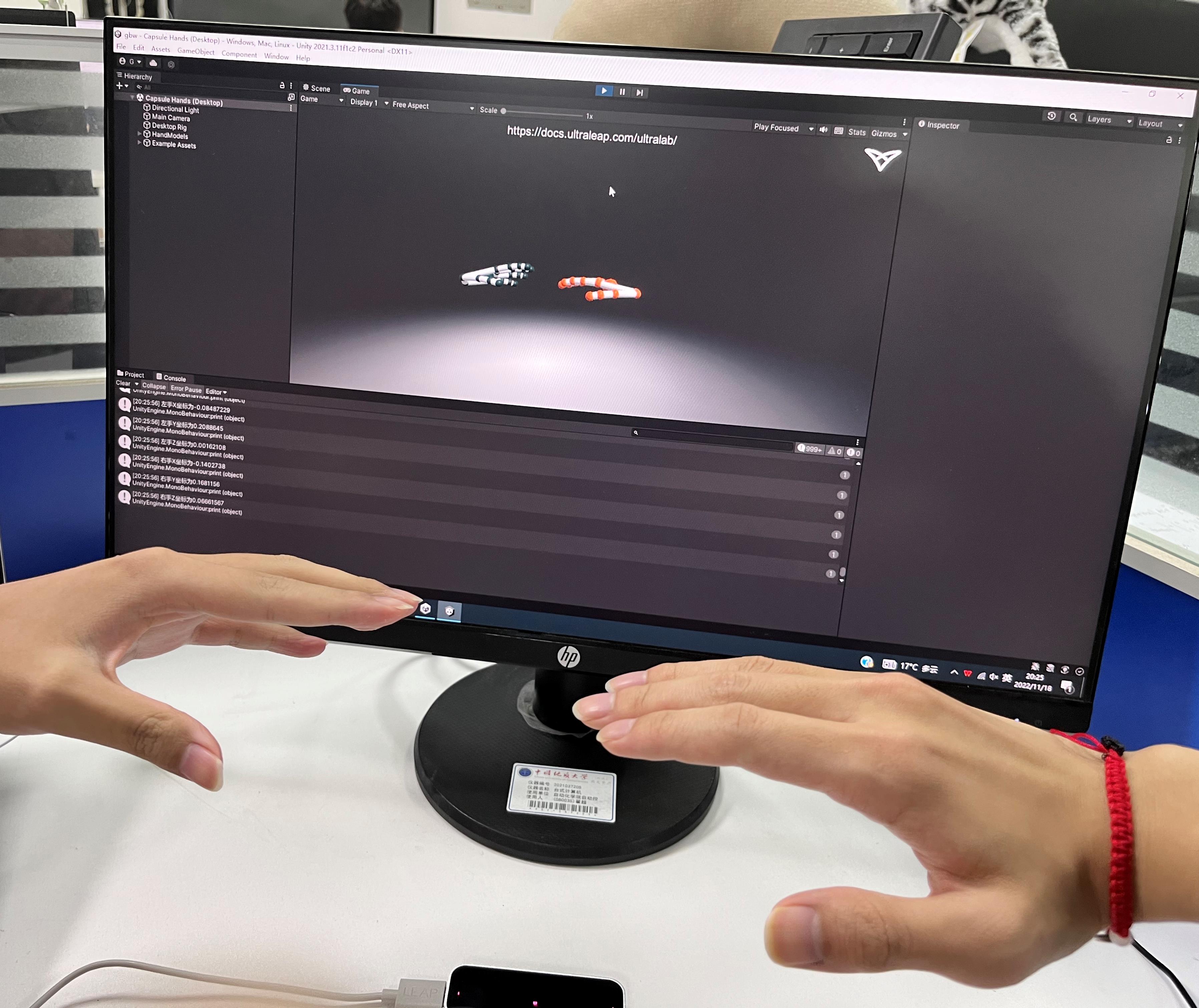}
\end{minipage}
\caption{Left: Participant uses the leap motion device to play a 2D mirror game with VP on the screen. Right: Two participants use leap motion device to play a 2D mirror game in the unity environment.}
\label{fig4}
\end{figure}

\subsection{Metrics}\label{sec:met}
A set of metrics are used to quantify the coordination performance between two players in the 2D mirror game. The RMSE of the position-time series represents the temporal correspondence between two players. The CV represents the coordination level between two players. The SVM characterizes the maximum range of movement. And their formal descriptions are presented as follows.
\begin{itemize}
\item [1)]
${\rm RMSE}=\sqrt {\frac{1}{n}\sum_{j=1}^{n}[(x_{1,j}-x_{2,j})^2+(y_{1,j}-y_{2,j})^2]}$, where $n$ is the number of sampling steps. In addition, $x_{1,j}$ and $x_{2,j}$ ($y_{1,j}$ and $y_{2,j}$) denote the positions of two players at the $j$-th step in the $x$-axis ($y$-axis) direction, respectively. A lower value of $RMSE$ indicates the better temporal correspondence~\cite{Zhai18a}.
\item [2)]
${\rm CV}=\|\frac{1}{n}\sum_{j=1}^{n}e^{i\delta\Phi_j}\|_2\in [0,1]$, where $\delta \Phi_j$ represents the sum of relative phase  between two players at the $j$-th step,  and $\|\cdot\|_2$ denotes the 2-norm. A higher value of $CV$ represents a higher coordination level between two players~\cite{Kre07}.
\item [3)]
${\rm SVM}=\|\mathbf{x}\|_{\infty}\|\mathbf{y}\|_{\infty}$, where $\mathbf{x}$ and $\mathbf{y}$ denote the position sequences in the $x$-axis and $y$-axis directions. $\|\cdot\|_{\infty}$ denotes the $\infty$-norm. A higher value of SVM indicates the better motion priming and novelty.
\end{itemize}

\subsection{Statistical analysis}\label{sec:sta}
In order to verify the effectiveness of ILC strategy, it is necessary to match the performance indexes of HP-VP pair with those of HP-HP pair (i.e., benchmark) as shown in Fig.~\ref{fig45}. Two HPs play the 2D mirror game and  their motion trajectories are recorded to compute the performance indexes as the benchmark. Then, the VP driven by the control strategy is adopted to replace HP2 and play 2D mirror game with HP1. The solo motions of HP2 are injected into VP so as to make VP act more like HP2 during interaction. Then the mean values and standard deviations of RMSE, CV and SVM are computed for each dyad. The matching results of the above control strategies are obtained with respect to benchmark. The indexes RMSE and CV are used to test the temporal correspondence and coordination level of each dyad of participants, while the metric SVM is employed to quantify the motion priming and novelty. The control parameters in the model are picked as follows $\alpha=0.01$, $\beta=0.01$, $\gamma=0.01$, $\omega=0.02$. Figure \ref{fig5} shows the 2D interactive motion of HP-VP pairs with ILC. It can be observed that VP and HP draw $"\infty"$ around a center.

\begin{figure}
\scalebox{0.3}[0.3]{\includegraphics{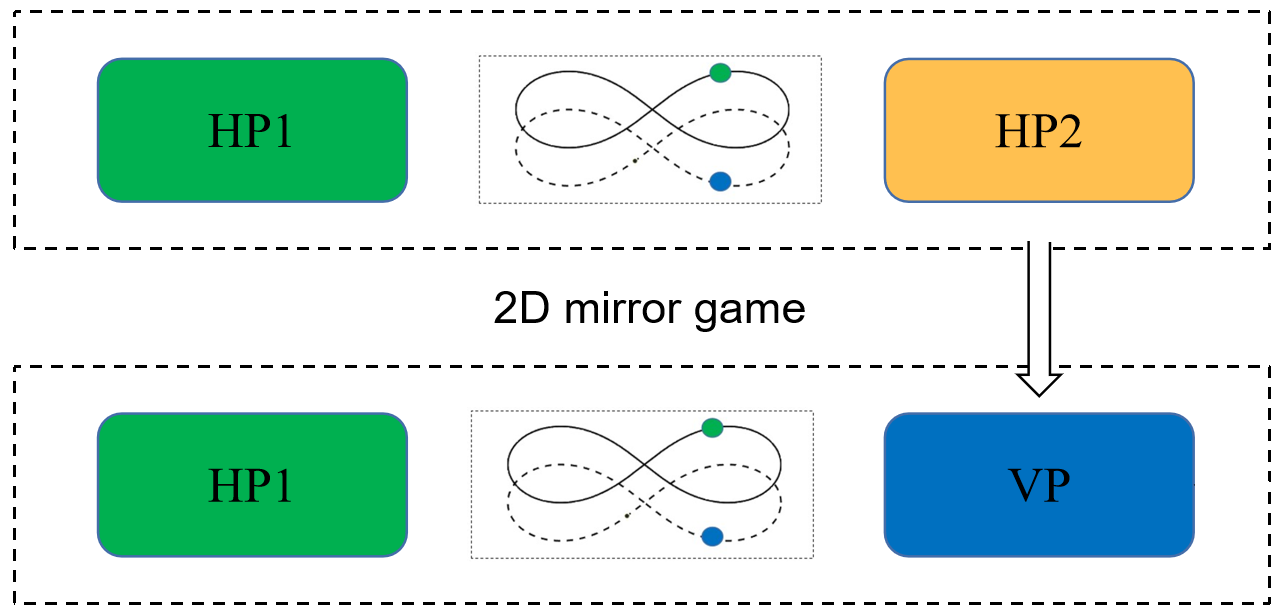}}\centering
\caption{Using HP-HP pair as benchmark to quantify the performance of VP driven by different control strategies.}\label{fig45}
\end{figure}

\begin{figure}[h]
\scalebox{0.3}[0.3]{\includegraphics{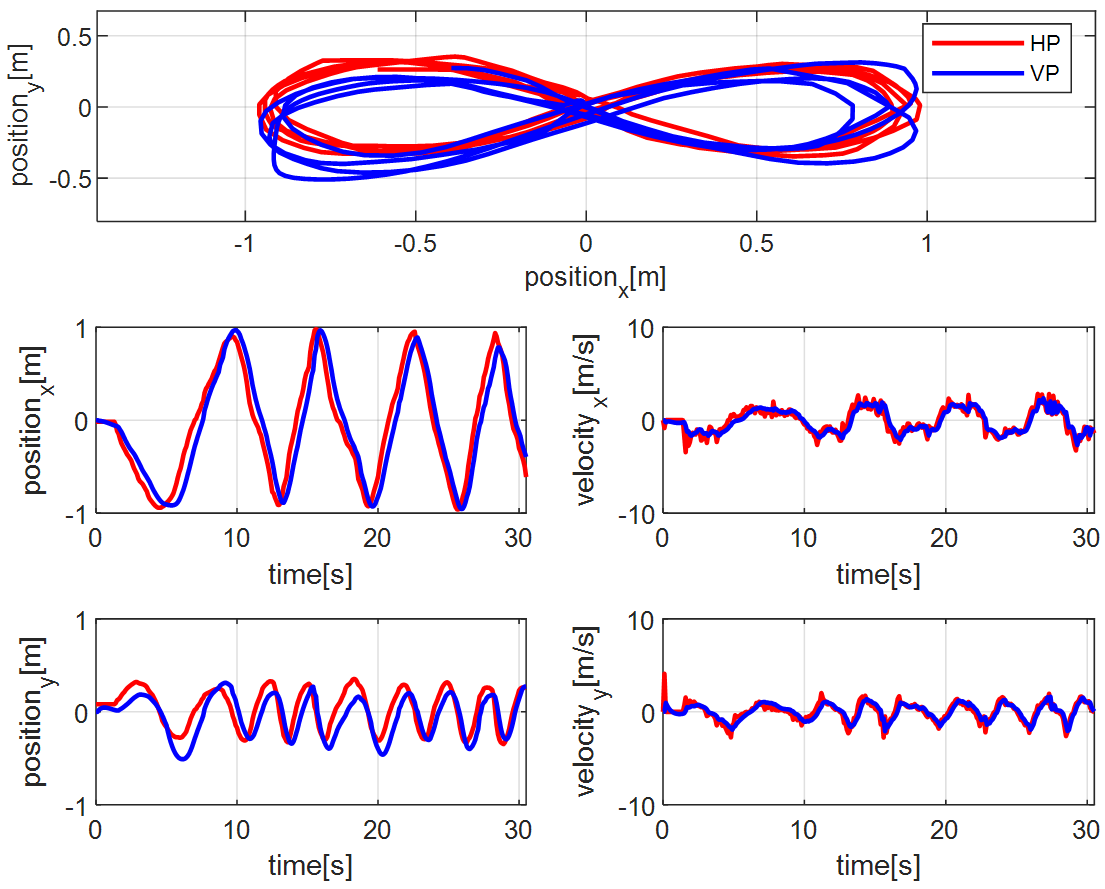}}\centering
\caption{2D interactive motion of HP-VP pairs with ILC.}\label{fig5}
\end{figure}

In the experiments, 4 Dyads are taken into account with each dyad including a HP-HP pair and a VP-HP pair.
Specifically, three trials are conducted for each HP-HP pair. Then the mean
values and standard deviations of RMSE, CV and SVM are computed for
the trials. The above values are regarded as the benchmark for the
validation of VP-HP interaction. Next, the model parameters are tuned for each VP-HP pair to match the benchmark of the HP-HP pair.  Essentially, each VP is assigned a different kinematic feature in advance to act as his/her avatar in the simulation.
Figures~\ref{fig6},~\ref{fig7} and~\ref{fig8} present matching results of VP-HP pairs with respect to their benchmarks (i.e., HP-HP pairs) in leader-follower condition. Specifically, RMSE in each dyad with ILC are $0.17\pm0.005$, $0.16\pm0.002$, $0.12\pm0.002$, $0.14\pm0.003$, respectively. It is observed that dyads 1, 2, and 4 with ILC can maintain a low level of RMSE and a minimum mismatch with the benchmark.
In terms of CV, the values with ILC are $0.43\pm0.002$, $0.15\pm0.002$, $0.23\pm0.001$, $0.39\pm0.003$, respectively.
Dyads 1, 2, and 4 with ILC obtain the minimum mismatch with the benchmark, which indicates a high coordination level.
Similarly, the values of SVM with ILC are $0.01\pm0.005$, $0.08\pm0.002$, $0.04\pm0.013$, $0.06\pm0.004$, respectively.
Note that the mismatch of ILC with the benchmark is smaller than other control strategies, although it does not lead to the largest value of SVM. It can be seen that the statistical performance of ILC is the closest to that of HP-HP pair. Specifically, for RMSE, Dyads 1, 2, and 4 obtain the best matching results. For CV, Dyads 1, 2, and 4 obtain the best matching results with ILC. For SVM, Dyads 2, 3, and 4 obtain the best matching results with ILC.

\begin{figure}[t!]
\centering
\begin{minipage}[t]{0.48\textwidth}
\centering
\includegraphics[width=7cm,height=6cm]{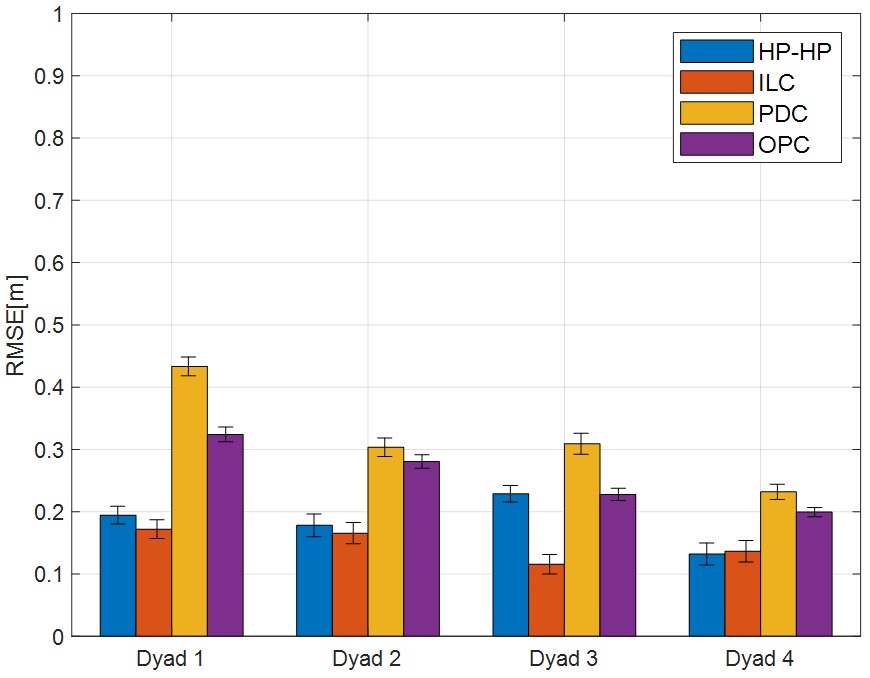}
\end{minipage}
\begin{minipage}[t]{0.48\textwidth}
\centering
\includegraphics[width=7cm,height=6cm]{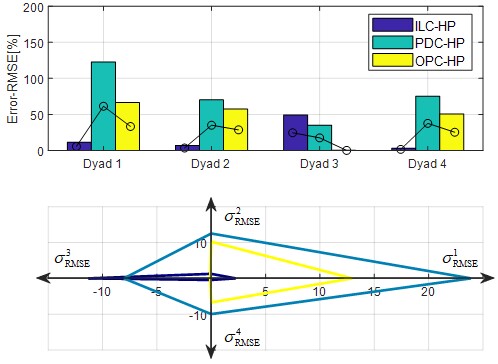}
\end{minipage}
\caption{\label{ireg} Statistical analysis of the LFC. Left: Mean values are represented by the
height of each rectangle, whereas standard deviations are represented by black
error bars. The blue rectangles and their error bars characterize the kinematic
feature of HP-HP pairs in terms of RMSE, which is viewed as the benchmark. The red rectangle represents the statistical performance of HP-VP pairs with ILC, the yellow rectangle represents the statistical performance of HP-VP pairs with  PD control (PDC) and the purple rectangle represents the statistical performance of HP-VP pairs with optimal control (OPC). Right: Error rate of RMSE.}
\label{fig6}
\end{figure}

\begin{figure}[t!]
\centering
\begin{minipage}[t]{0.48\textwidth}
\centering
\includegraphics[width=7cm,height=6cm]{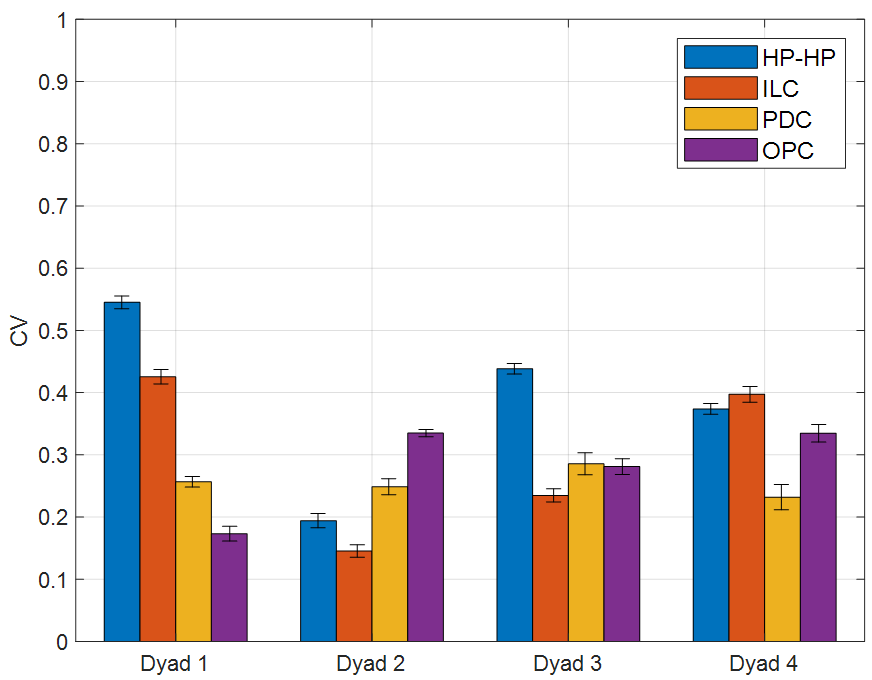}
\end{minipage}
\begin{minipage}[t]{0.48\textwidth}
\centering
\includegraphics[width=7cm,height=6cm]{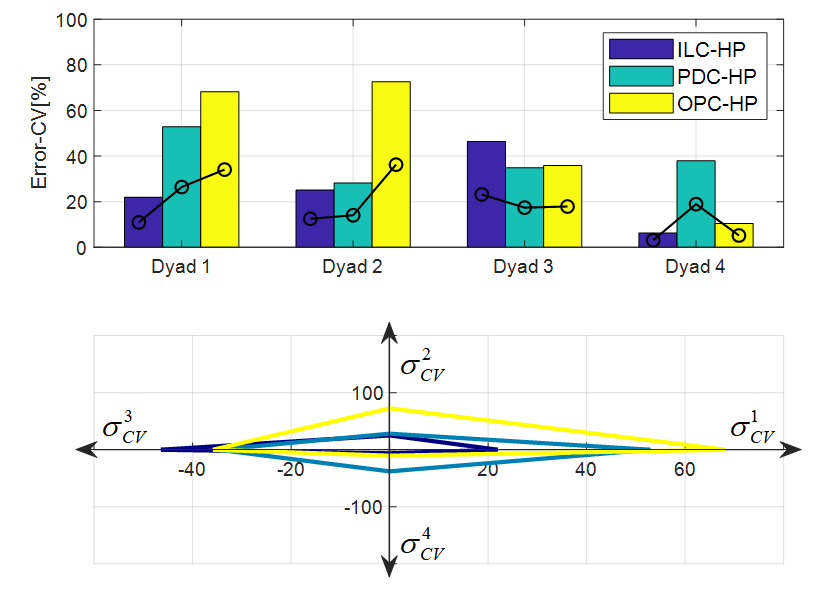}
\end{minipage}
\caption{Statistical analysis of the LFC. Left: Mean values are represented by the
height of each rectangle, whereas standard deviations are represented by black
error bars. The blue rectangles and their error bars characterize the kinematic
feature of HP-HP pairs in terms of CV, which is viewed as the benchmark. The red rectangle represents the statistical performance of HP-VP pairs with ILC, the yellow rectangle represents the statistical performance of HP-VP pairs with PDC and the purple rectangle represents the statistical performance of HP-VP pairs with OPC. Right: Error rate of CV.}
\label{fig7}
\end{figure}

\subsection{Comparisons with other control strategies} \label{sec:com}
In order to compare different control strategies, the error rates of the performance indexes are computed for the four dyads, respectively. And the whole performance of control strategies is quantified by the area of the polygon surrounded by the average error rates. Note that the parameters setting for PDC and OPC are the same as existing work~\cite{Zhai18b,Zhai18a}. For the ILC, the error rates of RMSE for 4 dyads are 11.52\%, 7.07\%, 49.39\% and 3.33\%, respectively. For the PDC, they are 122.83\%, 70.27\%, 35.14\% and 75.36\%, respectively.
In addition, the OPC leads to 66.68\%, 57.43\%, 0.44\% and 50.79\%, respectively. It is obvious that the area of the polygon with ILC is the smallest among three control strategies, which also applies to the metrics of CV and SVM. The ILC optimizes the control input iteratively based on state error to improve the control performance. It is demonstrated that ILC allows to achieve the best matching performance. According to experimental results, the performance of Dyads 1, 2 and 4 with ILC is relatively consistent with the RMSE benchmark. In terms of CV, Dyads 1, 2 and 4 with ILC match the benchmark well, and different control strategies lead to a similar performance in Dyad 3.

\begin{figure}[t!]
\centering
\begin{minipage}[t]{0.48\textwidth}
\centering
\includegraphics[width=7cm,height=6cm]{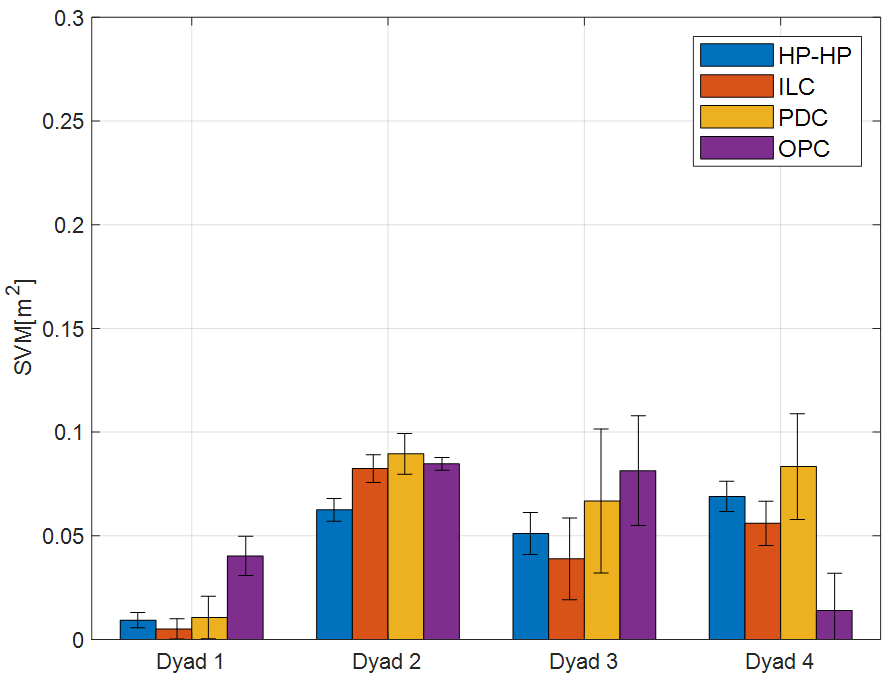}
\end{minipage}
\begin{minipage}[t]{0.48\textwidth}
\centering
\includegraphics[width=7cm,height=6cm]{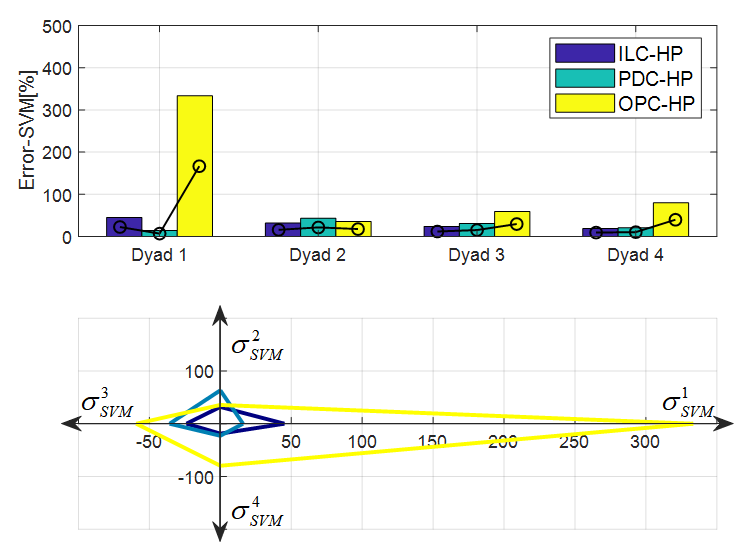}
\end{minipage}
\caption{Statistical analysis of the LFC. Left: Mean values are represented by the
height of each rectangle, whereas standard deviations are represented by black
error bars. The blue rectangles and their error bars characterize the kinematic
feature of HP-HP pairs in terms of SVM, which is viewed as the benchmark. The red rectangle represents the statistical performance of HP-VP pairs with ILC, the yellow rectangle represents the statistical performance of HP-VP pairs with PDC and the purple rectangle represents the statistical performance of HP-VP pairs with OPC. Right: Error rate of SVM.}
\label{fig8}
\end{figure}

\begin{figure}[t!]
\scalebox{0.3}[0.3]{\includegraphics{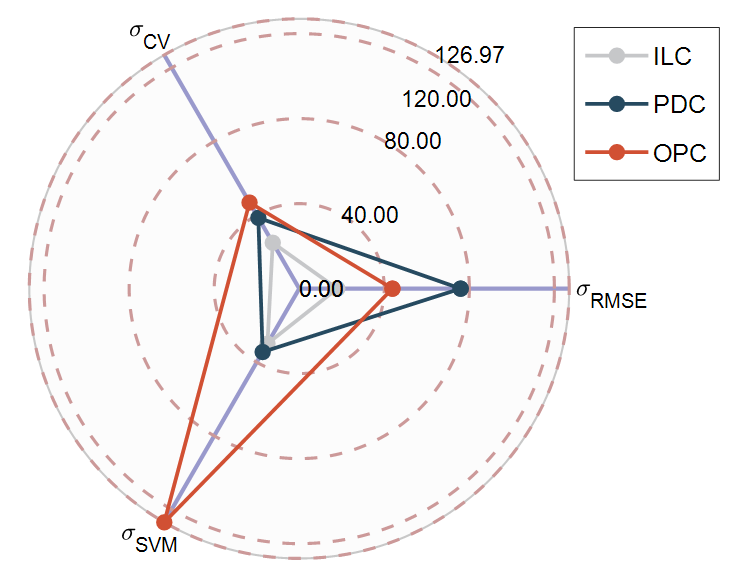}}\centering
\caption{\label{ireg}  Matching results of RMSE, CV and SVM. The area of the triangle in the radar chart shows the matching results of HP-HP experiment and HP-VP experiment. A smaller area of the triangle indicates the better matching result.}\label{figradar}
\end{figure}

In addition, the motion priming of VP to the HP is discussed with the aid of SVM as shown in Fig.~\ref{fig8}.
It is concluded that the ILC obtains the best performance among three strategies.
For the PDC and OPC, the inherent vibration of VP cannot be well suppressed. Therefore, the SVM of these control strategies are larger than that of ILC. But experimental results show that ILC can maintain a large range of motion continuously and stably. This might be because VP has the best tracking performance with ILC, so the HP motion can be tracked to the maximum and the range of VP motion is wider than other control methods. To demonstrate the advantage of ILC, Figure~\ref{figradar} shows the comprehensive matching results of all performance indexes. On the whole, ILC achieves the best comprehensive performance due to its smallest polygon area. A larger mismatch of CV may result from the unfamiliarity with "$\infty$" of participants at the initial stage.

\section{Conclusions}\label{sec:con}
This work developed a coordination model with ILC law in order to create a customized VP for social motor coordination with the HP in
the 2D mirror game. The proposed control law ensures the bounded state error while maintaining the desired kinematic feature of VP. Moreover, the advantages of ILC law were demonstrated through comparative analysis with other control strategies in the experimental validations.
Future work may include the study of optimal coordination mechanism between two players in the 2D mirror game to enhance the level
of motor coordination. Moreover, In the experimental validation, it is observed that the tracking performance of VP has been significantly improved after adding integral control or feedback control, which implies that integral control and feedback control can enhance the convergence of ILC. Therefore, future work may also include ILC with integral control and feedback control.

\section*{Acknowledgment}
This work is supported by the Fundamental Research Funds for the Central Universities, China University of Geosciences~(Wuhan), and it is also supported by the National Natural Science Foundation of China (62003088) and the National Natural Science Foundation of Fujian Province (2021J02008).

\section*{Appendix}

This section provides the proof of Theorem \ref{theo4.1}, together with parameter setting for experimental validations.

\subsection*{A. Proof of Theorem \ref{theo4.1}}\label{app_the1}
The sketch of proof is presented as follows. First of all, the relationship between state error and input error is established.
Afterwards, the iterative formula of input errors in two consecutive steps is created. Finally, the upper bound of state error is
estimated with the control law (\ref{ilc}). Before technical analysis, the time-weighted norm is defined as $\|\cdot\|_{\lambda}=\max_{t\in[0,T]}e^{-\lambda t}\|\cdot\|$ with the Frobenius norm $\|\cdot\|$.
The proof details are elaborated in following three steps.

\noindent\textbf{Step1:} In light of (\ref{ilc-sys}) and (\ref{error-sys}), one obtains
$\delta\dot{x}_k=\delta\bf{H}(\mathbf{x}_k,\bf{\xi})$+$\mathbf{B}\delta\mathbf{u}_k$.
By multiplying both sides of above equation with $\delta x_k$, one gets
$\big \langle \delta \dot x_k , \delta  x_k \big \rangle$=$\langle\delta H(\mathbf{x}_k,\bf{\xi})$, $\delta x_k\rangle$+$\big\langle B\delta u_k,\delta x_k\big\rangle$. It follows from (\ref{lipschitz}) that
\begin{equation*}
\begin{aligned}
\big\langle\delta\dot{x}_k,\delta{x}_k\big\rangle &=\big\langle\delta H(\bf{x}_k,\bf{\xi}),\delta {x}_k \big\rangle+\big\langle B\delta u_k ,\delta{x}_k \big\rangle \\
&\leq C_H\|\delta x_k \|^2+(B\delta u_k )^T\delta x_k \\
&\leq C_H\|\delta x_k \|^2+\|B\|\cdot \|\delta u_k \| \cdot \|\delta x_k \|.
\end{aligned}
\end{equation*}
Thus, one obtains
\begin{equation*}
\begin{aligned}
\frac{d[(\delta x_k )^T\cdot\delta x_k ]}{dt}=2\langle\delta \dot{x}_k,\delta x_k\rangle
&\le 2C_H\|\delta x_k \|^2+2\|B\| \cdot \|\delta u_k \| \cdot \|\delta x_k \| \\
&\le 2C_H\|\delta x_k \|^2+\|\delta x_k \|^2+\|B\|^2 \|\delta u_k \|^2 \\
&\le(2C_H+1)\|\delta x_k \|^2+\|B\|^2 \|\delta u_k \|^2,
\end{aligned}
\end{equation*}
which leads to
\begin{equation}\label{gron}
\begin{aligned}
\big \langle \delta  x_k, \delta x_k \big \rangle =\|\delta x_k \|^2\le \int_{0}^{t}\big[ (2C_H+1)\|\delta x_k(\tau)\|^2+\|B\|^2 \|\delta u_k(\tau)\|^2 \big] d\tau.
\end{aligned}
\end{equation}
Let $F =\int_{0}^{t}(2C_H+1)\|\delta x_k(\tau)\|^2d\tau$, $G=\int_{0}^{t}\|B\|^2 \|\delta u_k(\tau)\|^2 d\tau$ and the time derivative of $F$ can be estimated by $\dot F=(2C_H+1)\|\delta x_k \|^2\le (2C_H+1)F+ (2C_H+1)G$, which is equivalent to $\dot{F}-(2C_H+1)F \le (2C_H+1)G$. Thus, one has
\begin{equation*}
\frac{d}{dt}[F e^{-\int_{0}^{t}2C_H+1d\tau}] \le (2C_H+1)G e^{-\int_{0}^{t}2C_H+1d\tau},
\end{equation*}
which leads to
\begin{equation*}
F  \le  \int_{0}^{t} (2C_H+1)G(\tau)e^{\int_{\tau}^{t}2C_H+1ds}d\tau.
\end{equation*}
Then it follows from mean value theorems for definite integrals that
\begin{equation}\label{gronwallinequality}
\begin{aligned}
\|\delta x_k \|^2 &\le G+\int_{0}^{t} (2C_H+1)G(\tau)e^{\int_{\tau}^{t}2C_H+1ds}d\tau\\
&=G+G(\varepsilon)\int_{0}^{t} (2C_H+1)e^{\int_{\tau}^{t}2C_H+1ds}d\tau, \quad \varepsilon \in [0,t].
\end{aligned}
\end{equation}
Since $\int_{0}^{\varepsilon}\|B\|^2 \|\delta u_k(\tau)\|^2 d\tau \le \int_{0}^{\varepsilon}\|B\|^2 \|\delta u_k(\tau)\|^2 d\tau+\int_{\varepsilon}^{t}\|B\|^2 \|\delta u_k(\tau)\|^2 d\tau$, one gets
\begin{equation*}
\begin{aligned}
\|\delta x_k \|^2 \le G\left(1+\int_{0}^{t} (2C_H+1)e^{\int_{\tau}^{t}2C_H+1ds}d\tau\right)=Ge^{(2C_H+1)t}.
\end{aligned}
\end{equation*}
Because of $G=\int_{0}^{t}\|B\|^2 \|\delta u_k(\tau)\|^2 d\tau$, one obtains
\begin{equation*}
\begin{aligned}
\|\delta x_k \|^2 &\le e^{(2C_H+1)t}\|B\|^2\int_{0}^{t}\|\delta u_k(\tau)\|^2d\tau\\
&\le e^{(2C_H+1)t}\|B\|^2\int_{0}^{t}e^{2\lambda \tau}d\tau\|\delta u_k\|^2_\lambda\\
&=e^{(2C_H+1)t}\|B\|^2\frac{e^{2\lambda t}-1}{2\lambda}\|\delta u_k\|^2_\lambda.
\end{aligned}
\end{equation*}
Based on the definition of $\lambda$-norm, one gets
\begin{equation}\label{g1}
\begin{aligned}
\|\delta x_k \|^2_\lambda=\Big[\max\limits_{t\in[0,T]} e^{-\lambda t}\|\delta x_k \|\Big]^2=\max\limits_{t\in[0,T]} e^{-2\lambda t}\|\delta x_k \|^2\le \frac{\|B\|^2e^{(2C_H+1)T}}{2\lambda}\|\delta u_k \|^2_\lambda.
\end{aligned}
\end{equation}
\noindent \textbf{Step2:} By introducing control law (\ref{ilc}), it is ensured that input error is bounded by selecting appropriate parameters. As a result, state error is bounded according to inequality (\ref{g1}).
It follows from
\begin{equation*}
\begin{aligned}
\delta u_{k+1} &=u_h -u_k +u_k -u_{k+1} \\
&=\delta u_k -\kappa_p e_k-\kappa_v \dot e_k-\kappa_s s_k\\
&=\big[ I-\kappa_vCB\big]\delta u_k-\kappa_p C\delta x_k-\kappa_v C\delta H(\mathbf{x_k} ,\bf{\xi})-\kappa_s s_k,\\
\end{aligned}
\end{equation*}
that $\|\delta u_{k+1}\|\le \|I-\kappa_vCB\|\cdot\|\delta u_k \|+\|\kappa_pC\|\cdot\|\delta x_k\|+\|\kappa_vC\|\cdot \|\delta H(\mathbf{x_k} ,\bf{\xi})\|+\kappa_s\|s_k\|$, which leads to
$\|\delta u_{k+1} \|^2\le 4\big| \big|I-\kappa_vCB\|^2 \| \delta u_k \|^2+4\|\kappa_pC\|^2 \| \delta x_k \|^2+4\|\kappa_vC\|^2\| \delta H(\mathbf{x_k} ,\bf{\xi})\|^2+4\kappa_s^2\|s_k\|^2$.
Since $H(\mathbf{x_k} ,\bf{\xi})$ is Lipschitz continuous, one gets
$\|H(\mathbf{x_h} ,\bf{\xi})-H(\mathbf{x_k} ,\bf{\xi})\|^2\le C_H^2\|x_h -x_k\|^2$, which allows to get
$\|\delta u_{k+1} \|^2\le 4\|I-\kappa_vCB\|^2\| \delta u_k \|^2+(4\|  \kappa_pC\|^2+4C_H^2\|\kappa_vC\|^2)\|\delta x_k\|^2+4\kappa_s^2\|s_k\|^2$ and
\begin{equation}\label{g2}
\begin{aligned}
\|\delta u_{k+1} \|^2_\lambda\le4\big| \big|I-\kappa_vCB\|^2\| \delta u_k \|^2_\lambda+(4\|  \kappa_pC\|^2+4C_H^2\|  \kappa_vC\|^2)\| \delta x_k \|^2_\lambda+4\kappa_s^2\|s_k\|^2_\lambda.
\end{aligned}
\end{equation}
By integrating inequality (\ref{g1}) with inequality (\ref{g2}), one gets
\begin{equation*}
\begin{aligned}
\|\delta u_{k+1} \|^2_\lambda &\le 4\big| \big|I-\kappa_vCB\|^2\| \delta u_k \|^2_\lambda+\frac{\|B\|^2e^{(2C_H+1)T}(4\|  \kappa_pC\|^2+4C_H^2\|  \kappa_vC\|^2)}{2\lambda}\|\delta u_k \|^2_\lambda+4\kappa_s^2\|s_k\|^2_\lambda\\
&=(\sigma_1+\sigma_2)\|\delta u_k \|^2_\lambda+4\kappa_s^2\|s_k\|^2_\lambda
\end{aligned}
\end{equation*}
with
$$
\sigma_1=4\|I-\kappa_vCB\|^2, \quad\sigma_2=\frac{\|B\|^2e^{(2C_H+1)T}(4\|  \kappa_pC\|^2+4C_H^2\|  \kappa_vC\|^2)}{2\lambda}.
$$
Since $s_k=v-y_k=v-y_h+y_h-y_k=v-y_h+C\delta x_k$, one has
\begin{equation*}
\begin{aligned}
\|s_k\|^2_\lambda &\le 2\|v-y_h\|^2_\lambda+2\|C\|^2\|\delta x_k\|^2_\lambda \le 2\|v-y_h\|^2_\lambda+\frac{\|C\|^2\|B\|^2e^{(2C_H+1)T}}{\lambda}\|\delta u_k\|^2_\lambda,
\end{aligned}
\end{equation*}
which leads to $\|\delta u_{k+1}\|^2_\lambda\le(\sigma_1+\sigma_2+\sigma_3)\|\delta u_k \|^2_\lambda+8\kappa_s^2\|v-y_h\|^2_\lambda$
with $\sigma_3={4\|C\|^2\|B\|^2\kappa_s^2e^{(2C_H+1)T}}/{\lambda}$.
By introducing $\sigma=\sigma_1+\sigma_2+\sigma_3$ and $\eta=8\kappa_s^2\|v-y_h\|^2_\lambda$, one gets $\|\delta u_{k+1} \|^2_\lambda \le \sigma \|\delta u_k \|^2_\lambda+\eta$.

\noindent \textbf{Step3:} By tuning the parameters to ensure $\sigma<1$, it follows that
\begin{equation*}
\begin{aligned}
\|\delta u_{k+n}\|^2_\lambda\le \sigma \|\delta u_{k+n-1} \|^2_\lambda+\eta\le\sigma^n\|\delta u_k \|^2_\lambda+\eta\sum_{i=0}^{n-1}\sigma^{i} =\sigma^n\|\delta u_k \|^2_\lambda+\frac{\eta (1-\sigma^n)}{1-\sigma}.
\end{aligned}
\end{equation*}
Finally, it is concluded that $\lim\limits_{n\to\infty}\|\delta u_{k+n} \|^2_{\lambda}\le\eta/(1-\sigma)$.
According to (\ref{g1}), it is straightforward to obtain
$$
\lim\limits_{n\to\infty}\|\delta x_{k+n}\|^2_\lambda \le \frac{\eta\|B\|^2e^{(2C_H+1)T}}{2\lambda(1-\sigma)}.
$$
Note that $v$ and $y_h$ are bounded and $\kappa_s$ is a positive constant, which implies that $\delta x_k$ is bounded.
The proof is thus completed.


\subsection*{B. Parameter setting}

\hspace*{\fill}
\begin{table}[h] \centering
\caption{\label{ireg} Parameter setting for HP-VP pairs.}
\begin{tabular}{cccccc}
   \toprule
   Dyad ID & $\kappa_p$ & $\kappa_v$ & $\kappa_s$ \\
   \midrule
   Dyad 1 & 0.31 & 0.01 & 0.02  \\
   Dyad 2 & 0.45 & 0.02 & 0.03  \\
   Dyad 3 & 0.16 & 0.02 & 0.01  \\
   Dyad 4 & 0.41 & 0.04 & 0.03  \\
   \bottomrule
\end{tabular}
\end{table}


\begin{thebibliography}{99}

\bibitem{Bor07} Boraston, Z., Blakemore, S. J., Chilvers, R., \& Skuse, D. H. (2007). Impaired sadness recognition is linked to social interaction deficit in autism. \textit{Neuropsychologia}, 45, 1501-1510.

\bibitem{Cou06} Couture, S. M., Penn, D. L., \& Roberts, D. L. (2006). The functional significance of social cognition in schizophrenia: a review. \textit{Schizophrenia Bulletin}, 32(1), 44-63 .

\bibitem{Fol82} Folkes, V. S. (1982). Forming relationships and the matching hypothesis. \textit{Personality and Social Psychology Bulletin}, 8, 631 - 636.

\bibitem{Sch14} Schmidt, R. C., \& Fitzpatrick, P. A. (2014). Understanding the motor dynamics of interpersonal interactions. \textit{2014 IEEE International Conference on Systems, Man, and Cybernetics (SMC) }, 760-764.

\bibitem{Pax17} Paxton, A., \& Dale, R. (2017). Interpersonal movement synchrony responds to high and low-level conversational constraints. \textit{Frontiers in Psychology}, 8, 1135.

\bibitem{Wal15} Walton, A., Richardson, M. J., Langland-Hassan, P., \& Chemero, A. (2015). Improvisation and the self-organization of multiple musical bodies. \textit{Frontiers in Psychology}, 6, 313.

\bibitem{Fen16} Feniger-Schaal, R., Noy, L., Hart, Y., Koren-Karie, N., Mayo, A., \& Alon, U. (2016). Would you like to play together? Adults attachment and the mirror game. \textit{Attachment \& Human Development}, 18, 33-45.

\bibitem{Raf15} Raffard, S., Salesse, R.N., Marin, L., Del-Monte, J., Schmidt, R.C., Varlet, M., Bardy, B.G., Boulenger, J.P., \& Capdevielle, D. (2015). Social priming enhances interpersonal synchronization and feeling of connectedness towards schizophrenia patients. \textit{Scientific Reports}, 5, 8156.

\bibitem{alte} http://www.euromov.eu/alterego/

\bibitem{shar} http://sharespace.eu

\bibitem{Zhai18b} Zhai, C., Chen, M., Alderisio, F., Uteshev, A., \& Bernardo, M. (2018). An interactive control architecture for interpersonal coordination in mirror game. \textit{Control Engineering Practice}, 80, 36-48.

\bibitem{Zhai21} Zhai, C., He, Y., \& Zhang, C. (2021). Design and validation of feedback controller for social motor coordination with time-varying delays. \textit{Control Engineering Practice}, 109, 104756.

\bibitem{Zhai18a} Zhai, C., Alderisio,  F., Slowiski,  P., Tsaneva-Atanasova, K., \& Bernardo, M. (2018). Design and validation of a virtual player for studying interpersonal coordination in the mirror game. \textit{IEEE Transactions on Cybernetics}, 48, 1018-1029.

\bibitem{plos16} Zhai, C., Alderisio, F., Slowiski, P., Tsaneva-Atanasova, K., \& Di Bernardo, M. (2016). Design of a virtual player for joint improvisation with humans in the mirror game. \textit{PloS one}, 11(4), e0154361.

\bibitem{Aul22} Auletta, F., Kallen, R., Bernardo, M., \& Richardson, M. (2023). Predicting and understanding human action decisions during skillful joint-action using supervised machine learning and explainable-AI, \textit{Scientific Reports}, 13, 1, 4992.

\bibitem{Lom18} Lombardi, M., Liuzza, D., \& Bernardo, M. (2018). Using learning to control artificial avatars in human motor coordination tasks. \textit{IEEE Transactions on Robotics}, 37, 2067-2082.

\bibitem{ste05} Stefan, K., Cohen, L.G., Duque, J., Mazzocchio, R., Celnik, P., Sawaki, L., Ungerleider, L. and Classen, J. (2005). Formation of a motor memory by action observation. \textit{Journal of Neuroscience}, 25(41), 9339-9346.


\bibitem{bra96} Brashers-Krug, T., Shadmehr, R., \& Bizzi, E. (1996). Consolidation in human motor memory. \textit{Nature}, 382(6588), 252-255.


\bibitem{Kas18} Kashi, S., \& Levy-Tzedek, S. (2018). Smooth leader or sharp follower? Playing the mirror game with a robot. \textit{Restorative Neurology and Neuroscience}, 36, 147-159.

\bibitem{Noy11} Noy, L., Dekel, E., \& Alon, U. (2011). The mirror game as a paradigm for studying the dynamics of two people improvising motion together. \textit{Proceedings of the National Academy of Sciences}, 108, 20947-20952.


\bibitem{Hak04} Haken, H., Kelso, J. A., \& Bunz, H. (2004). A theoretical model of phase transitions in human hand movements. \textit{Biological Cybernetics}, 51, 347-356.

\bibitem{Slo16} Slowiski, P., Zhai, C., Alderisio, F., Salesse, R., Gueugnon, M., Marin, L., Bardy, B., Bernardo, M., \& Tsaneva-Atanasova, K. (2016). Dynamic similarity promotes interpersonal coordination in joint action. \textit{Journal of the Royal Society Interface}, 13, 20151093.

\bibitem{bri06} Bristow, D. A., Tharayil, M., \& Alleyne, A. G. (2006). A survey of iterative learning control, \textit{IEEE Control Systems Magazine}, 26(3), 96-114.

\bibitem{Pin02} Pinsky, M. A. (2002). Introduction to Fourier Analysis and Wavelets. \textit{Graduate Studies in Mathematics}, 102.

\bibitem{Gro19} Gronwall, T. (1919). Note on the derivatives with respect to a parameter of the solutions of a system of differential equations. \textit{Annals of Mathematics}, 20, 292.

\bibitem{gbw2} https://github.com/cug-gbw/gbw2022

\bibitem{Kre07} Kreuz, T., Mormann, F., Andrzejak, R. G., Kraskov, A., Lehnertz, K., \& Grassberger, P. (2007). Measuring synchronization in coupled model systems: a comparison of different approaches. \textit{Physica D: Nonlinear Phenomena}, 225, 29-42.

\end{thebibliography}
\end{document}